\documentclass[11pt]{article}

\usepackage[letterpaper,margin=1in]{geometry}
\usepackage[T1]{fontenc}
\usepackage[utf8]{inputenc}
\usepackage{lmodern}
\usepackage{microtype}
\usepackage{amsmath,amssymb,amsthm,mathtools}
\usepackage{booktabs,tabularx,array}
\usepackage{enumitem}
\usepackage[dvipsnames]{xcolor}
\usepackage{hyperref}

\hypersetup{
  colorlinks=true,
  linkcolor=MidnightBlue,
  citecolor=MidnightBlue,
  urlcolor=RoyalBlue,
  pdftitle={Fock-Space Geometry and Demailly's Conjecture},
  pdfauthor={Trung Hoa Dinh}
}

\setlist[itemize]{leftmargin=1.5em,itemsep=0.2em,topsep=0.35em}

\newtheorem{theorem}{Theorem}[section]
\newtheorem{proposition}[theorem]{Proposition}
\newtheorem{corollary}[theorem]{Corollary}

\theoremstyle{remark}

\newcommand{\PP}{\mathbb P}
\newcommand{\CC}{\mathbb C}

\newcommand{\Sym}{\operatorname{Sym}}
\newcommand{\whalpha}{\widehat{\alpha}}
\newcommand{\cD}{\mathcal D}
\newcommand{\cJ}{\mathcal J}
\newcommand{\HH}{\mathcal H}
\newcommand{\Fock}{\mathcal F_{+}}
\newcommand{\Ran}{\operatorname{Ran}}
\newcommand{\supp}{\operatorname{supp}}

\newcommand{\Tr}{\operatorname{Tr}}
\newcommand{\ket}[1]{\lvert #1\rangle}
\newcommand{\bra}[1]{\langle #1\rvert}
\newcommand{\braket}[2]{\langle #1\mid #2\rangle}
\newcommand{\ketbra}[2]{\lvert #1\rangle\!\langle #2\rvert}

\title{\bfseries Fock-Space Geometry and Demailly's Conjecture}
\author{Trung Hoa Dinh\\
\small Department of Mathematics and Statistics, Troy University\\
\small Troy, Alabama 36082, USA}
\date{September 2026}

\begin{document}
\maketitle

\begin{abstract}
We identify higher-order vanishing of homogeneous polynomials at finite projective point sets with a support--kernel problem in bosonic Fock space.  Under the Fischer--Bargmann realization, polynomial degree becomes particle number, while vanishing to a prescribed order becomes orthogonality to coherent jet excitations.  The least degree admitting a nonzero polynomial with the prescribed multiplicities is therefore the first bosonic sector in which a positive constraint operator acquires a dark state.  We then give an independent proof of Demailly's inequality by amplifying the last full-support certificate in positive characteristic and lifting the resulting rank statement back to characteristic zero.  The argument also identifies the dimensional correction in Demailly's bound with the bounded residue in the Frobenius decomposition.
\end{abstract}

\medskip
\noindent\textbf{Keywords.} Demailly conjecture; symbolic powers; Waldschmidt constant; bosonic Fock space; dark states; coherent jets; apolarity; Frobenius amplification.

\noindent\textbf{MSC 2020.} 14N20, 13A35, 13A15; secondary 81R30, 81Q12.

\section{From projective vanishing to bosonic Fock space}

Demailly's conjecture is a statement about the asymptotic cost of forcing a homogeneous polynomial to vanish deeply at finitely many projective points.  At first sight this belongs entirely to algebraic geometry and commutative algebra.  There is, however, a second description in which the same graded spaces become fixed-particle-number sectors of a bosonic Fock space and the same vanishing conditions become orthogonality constraints.  The purpose of this note is not to use quantum terminology as an analogy, but to isolate the exact Hilbert-space structure that is already present in the polynomial algebra and then use that structure to organize the proof.

The gain from this formulation is structural rather than numerical.  Classical symbolic vanishing asks whether a graded piece is zero; the Fock-space realization packages the same information into the support geometry of a positive synthesis operator.  Its range records the directions visible to the prescribed jets, while its kernel is exactly the surviving symbolic information.  This separates the geometry producing the constraints from the finite-dimensional information they carry.  At any fixed particle number, even a large or continuously parametrized family of constraints is represented by a finite spanning frame inside $\HH_d$.  The resulting support--kernel language is therefore not intrinsically tied to reduced point sets, although the theorem proved here is deliberately restricted to that clean case.  In this sense the quantum formulation does more than rename apolarity: it exposes a finite-information geometry that suggests a broader framework for constrained graded families.

Let $\HH\simeq\CC^{n+1}$ be a one-particle Hilbert space with orthonormal mode basis $\{\ket{e_0},\ldots,\ket{e_n}\}$.  Its symmetric bosonic Fock space is
\[
  \Fock(\HH)=\bigoplus_{d\ge0}\HH_d,
  \qquad
  \HH_d=\Sym^d(\HH).
\]
The number operator $\mathsf N=\sum_j a_j^\dagger a_j$ acts as $dI$ on $\HH_d$.  Thus the grading by homogeneous degree is exactly the grading by total boson number.  In particular, degree should not be confused with mixedness: a normalized vector in $\HH_d$ is still a pure $d$-boson state.  Mixed states will enter later, only after the local constraints are aggregated into a positive operator.

Put $S=\CC[y_0,\ldots,y_n]$ and equip $S_d$ with the Fischer inner product
\[
  \langle y^\alpha,y^\beta\rangle_{\mathrm F}=\alpha!\,\delta_{\alpha\beta}.
\]
The unitary map
\[
  U_d:S_d\longrightarrow\HH_d,
  \qquad
  U_d\!\left(\frac{y^\alpha}{\sqrt{\alpha!}}\right)=\ket{\alpha}
\]
identifies normalized monomials with occupation-number states.  Under $U=\bigoplus_dU_d$,
\[
  U_{d+1}M_{y_j}U_d^{-1}=a_j^\dagger,
  \qquad
  U_{d-1}\frac{\partial}{\partial y_j}U_d^{-1}=a_j.
\]
Hence multiplication is creation and differentiation is annihilation.  This is the finite-degree Fischer--Bargmann realization \cite{Bargmann1961}.

If $\ket\psi=\sum_j\psi_j\ket{e_j}$ is normalized, set $L_\psi(y)=\sum_j\psi_jy_j$ and $a^\dagger(\psi)=\sum_j\psi_ja_j^\dagger$.  Then
\[
  \ket{\psi;d}
  :=\frac{(a^\dagger(\psi))^d}{\sqrt{d!}}\ket\Omega
  =\ket\psi^{\otimes d}
\]
is the symmetric product state with all $d$ bosons occupying the same one-particle mode, and
\[
  U_d\!\left(\frac{L_\psi^d}{\sqrt{d!}}\right)=\ket{\psi;d}.
\]
Projectively, $[\psi]\mapsto[\psi;d]$ is exactly the degree-$d$ Veronese embedding.  General forms in $S_d$ correspond to general symmetric $d$-boson vectors; the Veronese locus is the much smaller locus of completely condensed product states.  This distinction is useful: the projective points at which we impose vanishing are coherent product directions, while the polynomial we are testing may represent an arbitrary symmetric state.

To convert evaluation into a Hilbert-space overlap, let $R=\CC[x_0,\ldots,x_n]$ act on $S$ by apolar differentiation, $x_j\circ g=\partial g/\partial y_j$.  For
\(
 f=\sum_{|\alpha|=d}f_\alpha x^\alpha\in R_d
\)
set
\[
  f^\sharp(y)=\sum_{|\alpha|=d}\overline{f_\alpha}y^\alpha,
  \qquad
  \ket{\Phi_f}=U_d(f^\sharp).
\]
A direct occupation-basis calculation gives
\begin{equation}
  f(\psi)=\frac{1}{\sqrt{d!}}\braket{\Phi_f}{\psi;d}.
  \label{eq:amplitude}
\end{equation}
Thus evaluating a homogeneous polynomial is, up to a universal normalization, a transition amplitude against a coherent $d$-boson product state.

The higher-order statement is equally exact.  For $m\ge1$ and $d\ge m-1$, define the coherent jet sector
\[
  \cJ_d(\psi,m)
  :=\Ran\!\left[(a^\dagger(\psi))^{d-m+1}:\HH_{m-1}\longrightarrow\HH_d\right].
\]
It is the span of $\ket{\psi;d}$ and all perturbative derivatives of the Veronese map through total order $m-1$.  In coordinates with $\psi=e_0$, it is spanned by the occupation states having at most $m-1$ particles outside the condensate mode $e_0$.  Therefore order-$m$ vanishing does not mean that the state contains no information; it means that all information lies in the orthogonal complement of what the first $m-1$ coherent probes can see.

For projective points $P_i=[\psi_i]$ and positive multiplicities $m_i$, put
\[
  I(\mathbf m)=\bigcap_{i=1}^s I(P_i)^{m_i}.
\]
The Emsalem--Iarrobino inverse-system theorem \cite{EmsalemIarrobino1995}, followed by the Fischer--Bargmann unitary, gives the exact bridge
\begin{equation}
  f\in I(\mathbf m)_d
  \quad\Longleftrightarrow\quad
  \ket{\Phi_f}\perp\sum_{i=1}^s\cJ_d(\psi_i,m_i).
  \label{eq:bridge}
\end{equation}
Equivalently,
\[
  I(\mathbf m)_d^{\perp_{\mathrm{ap}}}
  =\sum_i L_{\psi_i}^{d-m_i+1}S_{m_i-1}.
\]
The left side is symbolic vanishing; the right side is the span of all locally visible jet information.

For a finite reduced set $Z=\{P_1,\ldots,P_s\}$ with ideal $I=I(Z)$, uniform multiplicity gives $I^{(m)}=\cap_iI(P_i)^m$.  Its initial degree
\[
  \alpha(I^{(m)})=\min\{d:(I^{(m)})_d\ne0\}
\]
is therefore the first particle number at which a nonzero state can hide simultaneously from all order-$<m$ coherent probes.  The Waldschmidt constant
\[
  \whalpha(I)=\lim_{r\to\infty}\frac{\alpha(I^{(r)})}{r}
\]
measures the asymptotic particle-number cost per unit of vanishing depth.  Demailly's conjecture asks for the finite-to-asymptotic bound
\begin{equation}
  \whalpha(I)\ge
  \frac{\alpha(I^{(m)})+n-1}{m+n-1}.
  \label{eq:demailly}
\end{equation}
The case $m=1$ is Chudnovsky's inequality \cite{Demailly1982,Chudnovsky1981}.

The conjecture was recently proved for arbitrary finite point sets over algebraically closed fields of characteristic zero by H\`a and Sivakumar \cite{HaSivakumar2026}.  The purpose here is different: to identify the exact Fock-space and support--kernel geometry behind symbolic vanishing and to give an independent rank-amplification proof in that language.  A broader coefficientwise Frobenius-dilation framework, including applications to powers of forms, is developed in joint work with T\`ai Huy H\`a, currently submitted \cite{DinhHa2026}.

\begin{center}
\footnotesize
\begin{tabularx}{0.98\textwidth}{@{}>{\raggedright\arraybackslash}p{0.25\textwidth}>{\raggedright\arraybackslash}p{0.3\textwidth}>{\raggedright\arraybackslash}X@{}}
\toprule
\textbf{Projective/algebraic object} & \textbf{Bosonic object} & \textbf{Meaning}\\
\midrule
$P=[\psi]\in\PP^n$ & pure state $\ketbra{\psi}{\psi}$ & one ray in the one-particle space\\
degree $d$ & $\HH_d=\Sym^d\HH$ & fixed total particle number\\
$L_\psi^d$ & $\ket\psi^{\otimes d}$ & Veronese/coherent product state\\
multiplication by $y_j$ & $a_j^\dagger$ & creation\\
differentiation $\partial_{y_j}$ & $a_j$ & annihilation\\
order-$m$ vanishing at $P$ & orthogonality to $\cJ_d(\psi,m)$ & invisibility to all coherent jets below order $m$\\
$(I^{(m)})_d$ & common dark subspace & simultaneous local invisibility\\
$\alpha(I^{(m)})$ & first dark particle number & birth threshold of a dark sector\\
$\whalpha(I)$ & asymptotic threshold per order & asymptotic cost of robust nulling\\
\bottomrule
\end{tabularx}
\end{center}

\section{The support--kernel layer: from pure states to mixed constraint information}

The preceding dictionary is a pure-state/Fock-space description.  There is a second level that clarifies what the word ``dark'' means and also explains where mixed states naturally enter.

For a multiplicity profile $\mathbf m=(m_1,\ldots,m_s)$ define local synthesis maps
\[
  \mathsf V_{i,d,\mathbf m}
  =(a^\dagger(\psi_i))^{d-m_i+1}:\HH_{m_i-1}\longrightarrow\HH_d
\]
and combine them into
\[
  \mathsf V_{d,\mathbf m}:\bigoplus_i\HH_{m_i-1}\longrightarrow\HH_d,
  \qquad
  \mathsf V_{d,\mathbf m}(\Phi_1,\ldots,\Phi_s)
  =\sum_i\mathsf V_{i,d,\mathbf m}\Phi_i.
\]
Its range is the visible jet space
\[
  \Ran\mathsf V_{d,\mathbf m}=\sum_i\cJ_d(\psi_i,m_i).
\]
The associated positive frame, or constraint, operator is
\begin{equation}
  \mathsf H_{d,\mathbf m}
  =\mathsf V_{d,\mathbf m}\mathsf V_{d,\mathbf m}^\dagger
  =\sum_i\mathsf V_{i,d,\mathbf m}\mathsf V_{i,d,\mathbf m}^\dagger\ge0.
  \label{eq:constraintH}
\end{equation}
For every $\ket\Phi\in\HH_d$,
\[
  \bra\Phi\mathsf H_{d,\mathbf m}\ket\Phi
  =\sum_i\|\mathsf V_{i,d,\mathbf m}^\dagger\ket\Phi\|^2.
\]
Consequently,
\begin{equation}
  \ker\mathsf H_{d,\mathbf m}
  =\left(\Ran\mathsf V_{d,\mathbf m}\right)^\perp
  =:\cD_d(\Psi,\mathbf m).
  \label{eq:darkkernel}
\end{equation}
By \eqref{eq:bridge}, $\cD_d(\Psi,\mathbf m)$ is exactly the Hilbert-space image of $I(\mathbf m)_d$.

\begin{proposition}[Support--kernel decomposition]
Assume $\mathsf H_{d,\mathbf m}\ne0$ and define
\[
  \rho_{d,\mathbf m}
  =\frac{\mathsf H_{d,\mathbf m}}{\Tr\mathsf H_{d,\mathbf m}}.
\]
Then $\rho_{d,\mathbf m}$ is a density operator on $\HH_d$ and
\[
  \supp\rho_{d,\mathbf m}=\sum_i\cJ_d(\psi_i,m_i),
  \qquad
  \ker\rho_{d,\mathbf m}=\cD_d(\Psi,\mathbf m).
\]
Hence
\[
  \HH_d=\supp\rho_{d,\mathbf m}\oplus\ker\rho_{d,\mathbf m}
\]
is the orthogonal decomposition into visible and dark information.
\end{proposition}

\begin{proof}
Positivity and unit trace are immediate.  For every finite-dimensional operator $V$, one has $\Ran(VV^\dagger)=\Ran V$ and $\ker(VV^\dagger)=\ker V^\dagger=(\Ran V)^\perp$.  Apply this to $V=\mathsf V_{d,\mathbf m}$.
\end{proof}

Thus the mixed-state level does not come from the polynomial degree.  It appears after one aggregates the local observation channels.  Up to normalization, $\mathsf H_{d,\mathbf m}$ is a sum of positive local contributions; after normalization, it is a density operator (mixed whenever its support has dimension greater than one) whose support records everything accessible to the prescribed jets.  A zero response does not mean zero information: it means the state lies in the kernel of the current observation mechanism.

For uniform multiplicity $m$, Demailly's initial degree has the exact support interpretation
\begin{equation}
  \alpha(I^{(m)})
  =\min\{d:\ker\rho_{d,m}\ne0\}.
  \label{eq:supportthreshold}
\end{equation}
Before this threshold the constraint state has full support; at the threshold a dark sector is born.  This is a support phase transition across particle-number sectors.

There is also a perturbative interpretation.  If $\ket{\psi(\varepsilon)}$ is a smooth normalized curve with $\psi(0)=\psi$ and $\ket\Phi\perp\cJ_d(\psi,m)$, then
\[
  A_\Phi(\varepsilon)=\braket{\Phi}{\psi(\varepsilon);d}
\]
has $A_\Phi^{(j)}(0)=0$ for $0\le j<m$.  Therefore
\[
  A_\Phi(\varepsilon)=O(\varepsilon^m),
  \qquad
  |A_\Phi(\varepsilon)|^2=O(|\varepsilon|^{2m}).
\]
Multiplicity is robust nulling: the amplitude is invisible through order $m-1$, and the corresponding rank-one detection probability is suppressed to order $2m$.

The support formulation is not an additional analogy placed on top of apolarity.  It is the same finite-dimensional duality written after a Hermitian identification: below the initial degree every direction is visible to the jet synthesis map, while at the initial degree a nonzero orthogonal complement appears.  The proof of Demailly's inequality can therefore begin from the last full-support sector and ask whether that finite rank certificate amplifies.

\section{No-dark-state amplification and Demailly's inequality}

The support picture makes the proof particularly transparent.  Fix a positive multiplicity profile $\mathbf u=(u_1,\ldots,u_s)$ and write
\[
  F_{\mathbf u}=\bigcap_iI(P_i)^{u_i},
  \qquad
  a=\alpha(F_{\mathbf u}),
  \qquad
  A=a-1.
\]
By definition there is no dark state at particle number $A$:
\[
  (F_{\mathbf u})_A=0.
\]
Set $r_i=a-u_i\ge0$.  The apolar bridge rotates the vanishing statement into the full-coverage identity
\begin{equation}
  S_A=\sum_i L_i^{r_i}S_{u_i-1},
  \label{eq:sourcecoverage}
\end{equation}
where $L_i=L_{\psi_i}$.  Equivalently, the source constraint Hamiltonian at particle number $A$ is strictly positive: every direction is visible.

The proof now asks whether this finite full-rank certificate can be amplified.  The answer is yes along one prime-power sequence.

\begin{theorem}[Bosonic no-dark-state amplification]
\label{thm:amplification}
There exists a prime $p$ such that, for every $q=p^e$ and every positive multiplicity profile $\mathbf v=(v_1,\ldots,v_s)$ satisfying
\[
  v_i\ge q u_i+(n-1)(q-1)
  \qquad\text{for all }i,
\]
one has
\begin{equation}
  \alpha\!\left(\bigcap_iI(P_i)^{v_i}\right)
  \ge q\,\alpha(F_{\mathbf u})+(n-1)(q-1).
  \label{eq:amplification}
\end{equation}
Equivalently, if the source constraint operator has no dark state at particle number $A=a-1$, then the amplified target constraints have no dark state at
\[
  D_q=qA+n(q-1).
\]
\end{theorem}

\begin{proof}
Equation~\eqref{eq:sourcecoverage} is the surjectivity of the finite multiplication map
\[
  \Phi_A:\bigoplus_i S_{u_i-1}\longrightarrow S_A,
  \qquad
  (h_i)_i\longmapsto\sum_iL_i^{r_i}h_i.
\]
Choose a maximal minor witnessing surjectivity.  Spread the finitely many coefficients occurring in this matrix over a finitely generated integral $\mathbb Z$-algebra and localize so that the chosen determinant remains invertible.  A standard finite-type specialization argument then gives a finite residue field $\kappa$ of some characteristic $p>0$ for which the reduced source map remains surjective.  This prime is fixed from now on.

Let $q=p^e$ and put
\[
  D_q=qA+n(q-1).
\]
Over $\kappa$ set
\[
  J=(L_1^{r_1},\ldots,L_s^{r_s}).
\]
The source coverage says $J_A=S_A$, hence $J_t=S_t$ for every $t\ge A$.  We claim
\begin{equation}
  (J^{[q]})_{D_q}=S_{D_q},
  \qquad
  J^{[q]}=(L_1^{qr_1},\ldots,L_s^{qr_s}).
  \label{eq:frobcoverage}
\end{equation}
Because $\kappa$ is finite and hence perfect, every $G\in S_{D_q}$ has a unique decomposition
\begin{equation}
  G=\sum_{0\le\lambda_j<q} h_\lambda^{\,q}y^\lambda.
  \label{eq:frobdecomp}
\end{equation}
For each nonzero homogeneous summand,
\[
  q\deg h_\lambda=D_q-|\lambda|.
\]
There are $n+1$ variables and $|\lambda|\le(n+1)(q-1)$, so
\[
  q\deg h_\lambda
  \ge qA+n(q-1)-(n+1)(q-1)
  =q(A-1)+1.
\]
Since $\deg h_\lambda$ is integral, $\deg h_\lambda\ge A$.  Hence $h_\lambda\in J$, so $h_\lambda^q\in J^{[q]}$.  Every summand of \eqref{eq:frobdecomp} therefore lies in $J^{[q]}$, proving \eqref{eq:frobcoverage}.

In occupation-number language, the same step is the componentwise decomposition
\[
  \alpha=q\beta+\lambda,
  \qquad 0\le\lambda_j<q.
\]
The block $q\beta$ scales, while $\lambda$ is a bounded occupation residue.  The worst possible residue has total size $(n+1)(q-1)$; the shift in $D_q$ is exactly what prevents this residue from pushing the scalable core below the original full-coverage degree.  This is an arithmetic decomposition in a positive-characteristic Bargmann model, not a physical cloning or quantum channel.

The equality \eqref{eq:frobcoverage} is again a finite surjectivity statement.  Some maximal minor of its multiplication matrix is nonzero over $\kappa$; therefore the corresponding integral determinant is nonzero before reduction and remains nonzero after returning to characteristic zero.  Thus the amplified coverage also holds over $\CC$.

Now put
\[
  G_{\mathbf v}=\bigcap_iI(P_i)^{v_i}.
\]
If $v_i>D_q$ for some $i$, then $(G_{\mathbf v})_{D_q}=0$ immediately.  Otherwise define
\[
  e_i=D_q-v_i+1.
\]
The assumed lower bound on $v_i$ gives
\[
  e_i\le q(a-u_i)=qr_i.
\]
Hence
\[
  L_i^{qr_i}S_{D_q-qr_i}
  \subseteq
  L_i^{e_i}S_{D_q-e_i}.
\]
Summing and using the amplified coverage yields
\[
  S_{D_q}=\sum_iL_i^{e_i}S_{D_q-e_i}.
\]
Apolarity rotates this back to
\[
  (G_{\mathbf v})_{D_q}=0.
\]
Therefore
\[
  \alpha(G_{\mathbf v})\ge D_q+1
  =qa+(n-1)(q-1),
\]
which is \eqref{eq:amplification}.
\end{proof}

The theorem isolates the role of each language.  Hilbert-space orthogonality turns symbolic vanishing into a support statement; Frobenius amplifies one full-rank certificate; a nonzero maximal minor transports only that finite rank datum back to characteristic zero.  Nothing requires a positive-characteristic version of ordinary differential apolarity.

\begin{corollary}[Demailly]
\label{cor:demailly}
Let $Z=\{P_1,\ldots,P_s\}\subset\PP^n_{\CC}$ be nonempty and finite, with ideal $I=I(Z)$.  Then for every $m\ge1$,
\[
  \whalpha(I)
  \ge
  \frac{\alpha(I^{(m)})+n-1}{m+n-1}.
\]
\end{corollary}

\begin{proof}
Fix $m$ and put $a=\alpha(I^{(m)})$.  Apply Theorem~\ref{thm:amplification} with $u_i=m$ and
\[
  v_q=qm+(n-1)(q-1).
\]
For every $q=p^e$ supplied by the theorem,
\[
  \frac{\alpha(I^{(v_q)})}{v_q}
  \ge
  \frac{qa+(n-1)(q-1)}{qm+(n-1)(q-1)}.
\]
The sequence $\alpha(I^{(t)})$ is subadditive because $I^{(r)}I^{(t)}\subseteq I^{(r+t)}$, so Fekete's lemma gives the Waldschmidt limit.  Letting $q\to\infty$ along the prime-power subsequence gives
\[
  \whalpha(I)
  \ge
  \frac{a+n-1}{m+n-1}.
\]
\end{proof}

The proof explains the otherwise mysterious correction $n-1$.  In the positive-characteristic occupation decomposition, there are $n+1$ residue coordinates, each bounded by $q-1$.  After converting between the last no-dark degree, multiplicity, and the first dark degree, the net asymptotic residue cost becomes exactly $n-1$.  Thus Demailly's correction is the remainder cost left after the scalable Frobenius block has been separated from the bounded part.

H\`a and Sivakumar recently proved Demailly's conjecture for arbitrary finite point sets over an algebraically closed field of characteristic zero \cite{HaSivakumar2026}.  Their proof also passes to positive characteristic but works directly with symbolic vanishing and Hasse derivatives.  The route above first rotates the problem by characteristic-zero apolarity into a finite spanning statement, reduces only a rank certificate, performs the Frobenius amplification on that dual side, and then lifts only a maximal minor.  A more general coefficientwise Frobenius-dilation framework and its applications to varieties of powers are developed in joint work with T\`ai Huy H\`a, currently submitted \cite{DinhHa2026}.  The present formulation is designed to make the projective, bosonic, and arithmetic mechanisms visible simultaneously.

\section{Outlook: the opposite support transition for powers of forms}

The same positive-operator language also clarifies why a different class of secant problems requires a different information mechanism.  For general $G_1,\ldots,G_s\in S_d=\Sym^dV^*$, Terracini's lemma leads to
\[
  \Phi_G:S_d^s\longrightarrow S_{kd},
  \qquad
  (H_i)_i\longmapsto\sum_i H_iG_i^{k-1}.
\]
If $\mathsf K_G=\Phi_G\Phi_G^*$, then
\[
  \ker\mathsf K_G
  =\left(\sum_iG_i^{k-1}S_d\right)^\perp,
\]
so endpoint nondefectivity is exactly the full-support condition $\ker\mathsf K_G=0$.  Thus this problem moves through the same support--kernel geometry as Demailly, but in the opposite direction: Demailly asks when a dark sector is born as particle number grows, whereas the secant endpoint asks when a dark cokernel disappears as tangent blocks are added.  The generic $k$-rank problem suggested by Ottaviani is one prominent instance \cite{LORS2018}.

The contrast is especially transparent for fourth powers, whose tangent factor is cubic.  Frobenius is effective in the Demailly proof because it suppresses mixed terms and leaves a scalable block plus a bounded residue.  Cubic polarization instead stores information in the mixed terms.  With $\omega=i$,
\begin{equation}
  \frac14\sum_{\ell=0}^3\omega^{-r\ell}(P+\omega^\ell Q)^3
  =\binom3rP^{3-r}Q^r,
  \qquad r=0,1,2,3,
  \label{eq:phasecycling}
\end{equation}
so phase cycling separates the pure and mixed coherence sectors exactly.  This suggests a complementary principle: some algebraic rank problems are helped by decoupling, while others require coherence-preserving recovery.  Developing that distinction for varieties of powers is a separate problem and is not used anywhere in the proof above.

\medskip
\noindent\textbf{Conclusion.}
The graded polynomial ring and bosonic Fock space are the same symmetric algebra viewed through different structures.  Once a Hermitian realization is fixed, symbolic vanishing becomes orthogonality to coherent jet sectors, and the aggregate jet map produces a positive constraint operator whose support is the visible span and whose kernel is the symbolic dark sector.  In this language $\alpha(I^{(m)})$ is the first particle number at which full support is lost, while the Waldschmidt constant is the asymptotic rate of that threshold.  Frobenius amplification propagates the last full-support certificate below the threshold, yielding Demailly's inequality and a structural interpretation of its dimensional correction.

\section*{Declaration of generative AI and AI-assisted technologies}
During the preparation of this work, the author used OpenAI ChatGPT (GPT-5.6 Sol) for exploratory reformulations, algebraic cross-checks, organization of the exposition, and language editing.  The author reviewed the manuscript and takes full responsibility for its content.

\end{document}